\documentclass{llncs}

\usepackage{thmtools}
\usepackage{thm-restate}

\usepackage{amsmath,amssymb}
\usepackage{xspace}
\usepackage{url}
\usepackage{hyperref}
\usepackage[dvipsnames]{xcolor}
\usepackage{algorithm}
\usepackage{algpseudocode}
\usepackage{appendix}
\usepackage{babel}
\usepackage[capitalise,noabbrev]{cleveref}
\usepackage{paralist}

\usepackage{multirow}
\usepackage{mdwlist}
\usepackage{tablefootnote}
\usepackage{lipsum}
\usepackage{booktabs}
\usepackage{thm-restate}
\usepackage{graphicx}
\usepackage{subcaption}
\usepackage[obeyFinal,textsize=footnotesize]{todonotes}

\hypersetup{colorlinks=true,allcolors=blue}

\newcommand{\ignore}[1]{}

\DeclareMathOperator{\poly}{poly}

\DeclareMathOperator{\OPT}{OPT}

\DeclareMathOperator{\dist}{dist}

\DeclareMathOperator{\diam}{diam}

\newcommand{\colnote}[3]{\textcolor{#1}{$\ll$\textsf{#2}$\gg$\marginpar{\tiny\bf #3}}}
\newcommand{\hubert}[1]{\colnote{red}{#1--Hubert}{H}}
\newcommand{\wangbo}[1]{\colnote{blue}{#1--Bo}{WB}}

\newcommand{\R}{\mathbb{R}}

\title{Fully Dynamic Algorithms for Euclidean Steiner Tree\thanks{Full version of this paper is available at \cite{arxivFullVersion}}}

\author{T-H. Hubert Chan\inst{1}\orcidID{0000-0002-8340-235X}\thanks{T-H. Hubert Chan
was partially supported by the Hong Kong RGC grant 17203122.} \and
	Gramoz Goranci\inst{2}\orcidID{0000-0002-9603-2255}\and
	Shaofeng H.-C. Jiang\inst{3}\orcidID{0000-0001-7972-827X}\and
	Bo Wang\inst{4}\orcidID{0000-0003-4232-0478}\and
	Quan Xue\inst{5}\orcidID{0009-0001-5939-7271}}

\authorrunning{Chan et al.}

\institute{The University of Hong Kong, Hong Kong, China \\ \email{hubert@cs.hku.hk} \and
	University of Vienna, Vienna, Austria \\ \email{gramoz.goranci@univie.ac.at}\and
	Peking University, Beijing, China \\ \email{shaofeng.jiang@pku.edu.cn} \and
	The University of Hong Kong, Hong Kong, China\\ \email{bwang@cs.hku.hk}\and
	The University of Hong Kong, Hong Kong, China\\ \email{qxue@cs.hku.hk}}

\begin{document}
	\maketitle
	\setcounter{footnote}{0}
	 
	\begin{abstract}
	The Euclidean Steiner tree problem asks
	to find a min-cost metric graph that connects a given set of  \emph{terminal} points $X$ in $\mathbb{R}^d$, possibly using points not in $X$ which are called Steiner points.
	Even though near-linear time $(1 + \epsilon)$-approximation was obtained in the offline setting in seminal works of Arora and Mitchell,
	efficient dynamic algorithms for Steiner tree is still open.
	We give the first algorithm that (implicitly) maintains a $(1 + \epsilon)$-approximate solution which is accessed via a set of tree traversal queries,
	subject to point insertion and deletions,
	with amortized update and query time $O(\poly\log n)$ with high probability.
	Our approach is based on an Arora-style geometric dynamic programming,
	and our main technical contribution is to maintain the DP subproblems in the dynamic setting efficiently.
	We also need to augment the DP subproblems to support the tree traversal queries. 
\end{abstract}

	\section{Introduction}
In the Euclidean Steiner tree problem, a set $X \subset \R^d$ of points  called terminals is given, and the goal is to find the minimum cost metric graph that connects all points in $X$.
The optimal solution is a tree, and this tree can also use points other than terminals, called \emph{Steiner points}. 
Euclidean Steiner tree is a fundamental problem in computational geometry and network optimization, with many applications in various fields.



In the offline setting,
Euclidean Steiner tree is known to be NP-hard \cite{garey1977complexity}, even when the terminals are restricted to lie on two parallel lines~\cite{DBLP:journals/siamdm/RubinsteinTW97}.
An immediate approximation algorithm is via minimum spanning tree (MST),
which gives $O(1)$-approximation, and MST can be computed in near-linear time in Euclidean $\mathbb{R}^d$.
For a better approximation,
seminal works by Mitchell~\cite{DBLP:journals/siamcomp/Mitchell99} and Arora~\cite{DBLP:journals/jacm/Arora98}
developed PTAS's, i.e., $(1 + \epsilon)$-approximate polynomial time algorithms for the Euclidean Steiner tree,
which are based on geometric dynamic programming techniques.

Due to its fundamental importance, Steiner tree has also been studied in the dynamic setting.
Dynamic algorithms are algorithms that can handle modifications to the input data efficiently (in terms of running time), such as insertions and deletions.
Specifically, the study of dynamic Steiner tree in a graph setting was first introduced by Imase and Waxman \cite{imase1991dynamic}. In this setting, the objective is to find a minimum-weight tree that connects all terminals for a given undirected graph with non-negative edge weights.
Results for similar settings were also studied in followup papers \cite{megow2016power,gu2013power,gupta2014online,lkacki2015power}.


Even though an $O(1)$-approximation may be obtained from dynamically maintaining MST efficiently~\cite{DBLP:journals/dcg/Eppstein95,DBLP:journals/jacm/HolmLT01},
we are not aware of any previous work that explicitly studies the dynamic algorithm for the Euclidean Steiner tree,
especially whether one can convert the Mitchell/Arora's PTAS into an efficient dynamic algorithm to still achieve $(1 + \epsilon)$-approximation.
Technically, the Euclidean setting can differ greatly from the graph setting.
In the former, updates typically involve point insertion/deletion, whereas in the latter, updates involve more localized edge insertion/deletion.

\subsection{Our Results}
In this paper, we investigate the Euclidean Steiner tree problem in the fully dynamic setting, where terminal points in $\mathbb{R}^d$ can be inserted and deleted in each time step.
We give the first fully dynamic algorithm that (implicitly) maintains $(1 + \epsilon)$-approximation for the Steiner tree in amortized polylogarithmic time.
In particular, the algorithm maintains a) the total weight and b) the root node of this implicit tree.
Moreover, it can also answer several tree traversal queries: a) report whether a given point is in the solution tree and b) return neighbors of a given point $p$ in the solution tree, i.e., its parent and the list of its children. 

\begin{theorem}[Informal; see \Cref{thm:main}\footnote{In this informal statement we consider $d$ and $\epsilon$ as constants and their dependence are hidden in the big-O.}]
	\label{thm:main-informal}
	There exists a fully-dynamic algorithm for Euclidean Steiner tree
	that for any $0 < \epsilon, \delta < 1$ and any operation sequence 
	handles the $t$-th update in amortized $O(\poly\log(t) \log {1 / \delta})$  
	time for every $t$,
	such that with probability at least $1 - \delta$,
	the maintained Steiner tree is $(1 + \epsilon)$-approximate at all time steps simultaneously.
	Furthermore, conditioning on the randomness of the maintained solution,
	the algorithms for accessing the maintained tree are deterministic.
\end{theorem}

Note that here we do not consider time cost for pre-computation (see Section~\ref{sec:dynamic} for the detailed discussion), i.e., constructing an empty solution, and $t$ is the upper bound of the cardinality of the active set (see \Cref{sec:prel} for a formal definition).

Notice that our algorithm can report a $(1 + \epsilon)$-apprximation to the optimal \emph{value}. This value estimation can be generalized to TSP and related problems, including all solvable problems mentioned in~\cite{DBLP:journals/jacm/Arora98} such as $k$-MST and matching.
This is due to the fact that the geometric dynamic programming is defined on a quadtree with height $O(\log n)$, and that an update only requires to change $O(\poly\log n)$ subproblems associated with a leaf-root path of length $O(\log n)$.
We give a detailed discussion of this in \Cref{sec:tech_overview}.

However, the same approach cannot readily yield an efficient algorithm for maintaining a \emph{solution} (which is a set of edges).
One outstanding issue is that even though only $\poly \log n$ subproblems change,
the solution that is deduced from the DP subproblems can change drastically.
We thus seek for a weaker guarantee, which still offers sufficient types of queries to traverse the implicit tree, instead of explicitly maintaining all edges.
Even though one may maintain the solution within $O(1)$-approximation
by maintaining MST in $\poly\log n$ time per update~\cite{DBLP:journals/dcg/Eppstein95,DBLP:journals/jacm/HolmLT01},
it is an open question to improve the ratio to $1 + \epsilon$.
Indeed, we remark that maintaining $(1 + \epsilon)$-approximate solutions for
other related problems, such as TSP, is also open, although it may be possible to achieve a weaker guarantee that can be used to reconstruct the solution via multiple queries,
similar to our \Cref{thm:main-informal}.

\subsection{Technical Overview}
\label{sec:tech_overview}
Conceptually, our algorithm is based on Arora's geometric DP~\cite{DBLP:journals/jacm/Arora98} which gives PTAS's for Euclidean Steiner tree and various related problems in the offline setting,
and our main technical contribution is to make it dynamic.
For simplicity, we expose the ideas only in 2D, and the input is situated in a discrete grid $[\Delta]^2$, where $\Delta$ is a parameter. 
We start with a brief review of Arora's approach.
%

\paragraph{\textbf{Review of Arora's Approach.}}
Let $n$ be the number of points in the input dataset (recalling here we are working with a static dataset), and again for simplicity we assume $\Delta = \poly(n)$.
To begin, the algorithm builds a randomly-shifted quadtree on $[2\Delta]^2$.
Specifically, the bounding square $[2\Delta]^2$ is evenly divided into four sub-squares of half the edge-length, and this process is done recursively, until reaching a unit square.
This recursive process naturally defines a tree whose nodes are the squares.
Then crucially, the entire quadtree is randomly shifted independently in each coordinate $i$ by a uniformly random value $v_i \in [-\Delta, 0]$.
Indeed, in~\cite{DBLP:journals/jacm/Arora98}, it was shown that over the randomness of the shift, with constant probability, a ``structured'' solution that is $(1 + \epsilon)$-approximate exists,
and finally, such a ``structured'' solution may be found by a dynamic program.
The dynamic program is somewhat involved, but at a high level, it defines $\poly \log n$ subproblems for each square in the quadtree, 
and each subproblem can be evaluated from the information of the $\poly\log n$ subproblems associated with the child squares.


\paragraph{\textbf{Making Arora's Approach Dynamic: Value Estimation.}} 
Our main observation is that, for a fixed quadtree, when a point $x$ is inserted or deleted, only subproblems that are associated with squares containing $x$ need to be updated/recomputed (and other subproblems do not change).
The number of levels is $O(\log \Delta) = O(\log n)$ and in each level only $O(1)$ squares are involved and need to be recomputed,
hence the number of subproblems to be re-evaluated is $\poly\log n$.
This observation directly leads to an efficient dynamic algorithm that maintains an approximation to the optimal Steiner tree.
Moreover, this is very general as it only uses the quadtree structure of DP,
hence, our result for value estimation also generalizes to other related problems such as TSP, particularly those solvable in~\cite{DBLP:journals/jacm/Arora98}.

\paragraph{\textbf{Other Queries.}}
However, in the abovementioned approach, the solution defined by the DP, which is a set of edges, is \emph{not} guaranteed to have a small change per update.
Instead, we maintain an implicit solution that supports the following operations: (i) querying the root, (ii) determining whether a point is Steiner point, and (iii) reporting for a given point $p$ its parent and a list of children.
To support these queries, we augment Arora's DP and include additional information in each subproblem.
In the original Arora's DP, a subproblem is defined with respect to a square $R$.
Intuitively, if one restricts a global solution to $R$, then it breaks the solution into several small subtrees.
Thus, in Arora's approach, the subproblems associated with $R$ also need to describe the connectivity of these subtrees.
Now, in our algorithm, we additionally designate the \emph{root} of each subtree, and encode this in the subproblem.
This information about the root is crucial for answering the several types of queries including the root query and children-list query.
Since we add this new information and requirement in the subproblems,
we need to modify Arora's DP, and eventually we show this can be maintained by enlarging the time complexity by only a $\poly\log n$ factor.




\subsection{Related Work}

\paragraph{\textbf{Steiner Tree in the Graph Setting.}}
The Steiner tree Problem is a well-known problem in graph theory that involves finding the minimum cost tree spanning a given set of vertices in a graph, i.e., given an undirected graph with edge weights and a subset of vertices, the objective of the problem is to find a tree that connects all of the specified vertices at the lowest possible cost. The tree can also include additional vertices (known as Steiner points) that are not part of the original subset. Steiner trees have been used in various applications, such as VLSI design of microchips \cite{held2011combinatorial}, multipoint routing \cite{imase1991dynamic}, transportation networks \cite{cheng2013steiner}, and phylogenetic tree reconstruction in biology \cite{hwang1992steiner}. However, the Steiner tree problem in general graphs is NP-hard \cite{hwang1992steiner} and cannot be approximated within a factor of $1 + \epsilon$ for sufficiently small $\epsilon > 0$ unless P = NP \cite{bern1989steiner}. Therefore, efficient approximation algorithms are sought instead of exact algorithms. A series of papers gradually improved the best approximation ratio achievable within polynomial time from $2$ to $1.39$ for the Steiner tree problem in general graphs \cite{matsuyama1980approximate,zelikovsky199311,berman1994improved,zelikovsky1996better,promel1997rnc,karpinski1997new,hougardy19991,robins2005tighter,byrka2010improved}.

\paragraph{\textbf{Fully-dynamic Algorithms for Steiner Tree.}}
As noted by \cite{cheng2013steiner}, the understanding the dynamic complexity of Steiner Tree problem is an important research question due to its practical applications in transportation and communication networks. In the dynamic setting, a sequence of updates would be made to an underlying point set, and the goal of the dynamic Steiner tree problem is to maintain a solution dynamically and efficiently for each update. Imase and Waxman \cite{imase1991dynamic} first introduced this problem, and various papers have explored this direction \cite{megow2016power,gu2013power,gupta2014online}. For general graphs, Łacki and Sankowski \cite{lkacki2015power} achieved the current best approximation ratio for the dynamic setting at $(6+\epsilon)$, utilizing a local search technique. Let $D$ denotes the stretch of
the metric induced by $G$, the time complexity for each addition or removal is $
\tilde{O}(\sqrt{n} \log D)$ (note that $n$ is the size of the terminal set). 
However, there is no polynomial time approximation scheme for the dynamic Steiner tree problem. 

	\section{Preliminaries}
\label{sec:prel}
\paragraph{\textbf{Notations.}}
For integer $n \geq 1$, let $[n] := \{1, \ldots, n\}$.
In Euclidean space $\mathbb{R}^d$,
a metric graph is a graph whose vertex set is a subset of $\mathbb{R}^d$,
and the edges are weighted by the Euclidean distance between the end points.
For any $x, y \in \mathbb{R}^d$, the distance between them is defined as $\dist(x, y) := \|x - y\|$.
For a point set $S \subset \mathbb{R}^d$,
let $\diam(S) := \max_{x, y \in S}{\dist(x, y)}$ be its diameter.
When we talk about graphs (e.g., tours, trees) we always refer to a metric graph.
For a (metric) graph $G=(V, E)$ where $V$ is the set of vertices and $E$ is the set of edges. The weight of an edge $(x, y)\in E$ is defined as $w(x, y)$. Define the weight of the graph $w(G)=\sum_{(x, y)\in E} w(x, y)$. (Note that in this paper, $w(x, y)=\dist(x, y)$.) 

When we talk about a $d$-dimension hypercube in $\mathbb{R}^d$, we consider both its boundary (its $2d$ facets) and interior. 
For a set of points $X$, we will use $\OPT(X)$ to denote the cost of the optimal Steiner tree with respect to $X$. 

\begin{definition}[Euclidean Steiner Tree]
	\label{def:steiner_tree}
	Given $n$ \emph{terminal} points $X \subset \mathbb{R}^d$, 
	find a min-weight graph that connects all points in $X$.
	In a solution of Steiner tree, non-terminal points are called Steiner points.
\end{definition}

\paragraph{\textbf{Implicit Steiner Tree and Queries.}}
As mentioned, our algorithm maintains an (implicit) $(1+\epsilon)$-approximate Steiner tree along with its cost and the root node. By implicit, we mean we do not maintain an explicit edge set; instead, we provide access to this tree indirectly via queries.
The algorithm supports the \emph{membership queries} and the \emph{neighbor queries}. The concrete definitions are given below.

\begin{itemize}
	\item Membership queries. {Given a point $u$, if $u$ is a node in the implicit Steiner tree, 
	 return ``Yes''; otherwise, return ``No''.}
	\item Neighbor queries. {Given a point $u$, if $u$ is not a node in the implicit Steiner tree, return ``Null'';
	if a point $u$ is a node in the underlying Steiner Tree, return its all neighbors in the implicit Steiner tree, i.e., the parent and the list of children of $u$. }
\end{itemize}

\paragraph{\textbf{Fully Dynamic Setting.}}
We study the Euclidean Steiner Tree problem in a dynamic setting. Suppose $X$ is the underlying dataset, and we assume $X \subseteq [\Delta]^d$ for some integer parameter $\Delta$. 
This is a natural assumption, since in the typical setting where the set $X$ has at most $n$ points and the aspect ratio of the distances is $\poly(n)$,
the dataset can be discretized and rescaled so that it fits into $[\Delta]^d$ grid for $\Delta = \poly(n)$.
For any time step $t\in\mathbb{N}$, an update operation $\sigma_t\in X\times \{+, -\}$ consists of a point $x_t\in \mathbb{R}^d$ and a flag indicating whether the update is an insertion ($+$) or a deletion ($-$).

Multiple updates can be performed on the same coordinate, involving the insertion or deletion of multiple points. We track a `multiplicity' value for each coordinate, initially set to 0. When a point is inserted, the coordinate's multiplicity increases by 1, and when a point is deleted, the multiplicity decreases by 1.
Let $X_t$ be the \emph{active set} at time step $t$, which contains all the co-ordinates (we will view them as points) that have non-zero multiplicity. 
Note that for Steiner tree, the multiplicity of points do not affect the cost (as they can be connected with $0$-length edges).
An update operation is \emph{valid} if a point to be deleted is in the active set. 
Queries will be made after each update.

\paragraph{\textbf{Amortized Time Cost.}}
Suppose $A_t$ is the cost for time step $t$, the amortized time cost for $t$ is the average time cost for the first $t$ steps, i.e., $\frac{1}{t}\sum_{i=1}^{t} A_i$.

	\subsection{Review of Arora's Approach~\cite{DBLP:journals/jacm/Arora98}}
\label{sec:rev}

\paragraph{Randomly-shifted Quadtree.}
Assume the input is in $[\Delta]^d$.
Without loss of generality, we assume $\Delta$ is a power of 2.
For each $i \in \{1, 2, \dots, d\}$,
pick $s_i$ uniformly at random in $[0, \Delta]$. 
Take the hypercube $R_{\log \Delta}$ as $[2\Delta]^d$ shifted by $- s_i$
in each coordinate $i$, and this still contains the dataset.
We evenly subdivide this hypercube into $2^d$ child clusters,
and we continue this process recursively until the hypercube is of diameter at most $1$.
This process naturally induces a $2^d$-ary tree , i.e., $R_{\log \Delta}$ is the root, and each non-leaf cluster has $2^d$ child clusters.
For $i \in \{0, \ldots, \log \Delta\}$, a level-$i$ cluster is a hypercube $R_i$ with side-length $2^i$. Denote $\mathcal{R}_i$ as the collection of all such 
hypercubes with side-length $2^i$.
Denote $\mathcal{R} := \cup_{i=0}^{\log \Delta} \mathcal{R}_i$.

\ignore{
	Let $R_0$ be the square of edge-length $2^L$ and lower-left corner coordinate
	$(-s_x, -s_y)$ where $s_x, s_y$ are independent uniform random variables on $[0, \Delta]$.
	Clearly, $X$ is contained in $R_0$.
	Subdivide $R_0$ recursively: each time evenly divide the current square into $4$
	squares of half the edge length, until side length of $0.5$ is reached
	(noting that each square of side-length $0.5$ contains at most one point from $X$).
	Denote this collection of squares by $\mathcal{R}$,
	and the squares of side-length $2^i$ are called level-$i$ squares.
}

\paragraph{$(m, r)$-Light Graph.} 
For every hypercube $R$ in $\mathcal{R}$, designate its $2^d$ corners
as well as $m$ evenly placed points on each of $R$'s $2d$ facets as portals. 
Specifically, portals in a facet is an orthogonal lattice in the $(d-1)$-dimension hypercube.
If the side length of the hypercube is $\ell$ and the spacing length between the portals is $\alpha\cdot \ell$, then we will have $m \leq (\frac{1}{\alpha}+1)^{d-1}$.
For example, in $2D$, a hypercube is a square, and portals in this case is $4m-4$ points evenly placed on the $4$ boundaries.
An $(m, r)$-light graph is a geometric graph such that for all $R$ in $\mathcal{R}$,
it crosses each facet in $R$ only via the $m$ portals for at most $r$ times.

\begin{theorem}[Structural Property~\cite{DBLP:journals/jacm/Arora98} in $\R^d$]
	\label{thm:struct}
	For every $0 < \epsilon < 1$ and terminal set $X$ in $\R^d$,
	there is a (random) Steiner tree solution $T'$ defined with respect to the randomly-shifted quadtree $\mathcal{R}$ of $X$,
	such that $T'$ is $(m := O(d^{\frac{3}{2}} \cdot \frac{\log \Delta}{\epsilon})^{d-1}, r := O(\frac{d^{2}}{\epsilon})^{d-1})$-light\footnote{We state the exact dependence of $d$ which was not accurately calculated in \cite{DBLP:journals/jacm/Arora98}, see the appendix for the details.}\label{ftnote:parameters} 
	and
	\begin{align*}
		\Pr[w(T') \leq (1 + \epsilon) \cdot \OPT] \geq \frac{1}{2}.
	\end{align*}
\end{theorem}

\paragraph{Dynamic Programming (DP).}
One can find the optimal $(m, r)$-light solution for Steiner tree in near-linear time
using dynamic programming (for $m, r$ guaranteed by Theorem~\ref{thm:struct}).
In the original DP for Steiner tree~\cite{DBLP:journals/jacm/Arora98}, 
a DP entry is indexed by a tuple $(R, A, \Pi)$, where
\begin{itemize}
	\item $R$ is a hypercube in $\mathcal{R}$,
	\item $A$ is a subset of the \emph{active} portals of $R$ such that
	$A$ contains at most $r$ portals from each facet in $R$,
	\item $\Pi$ is a partition of $A$.
\end{itemize}
If $R$ is a leaf node, then $|A| \leq 1$.  In the case
that $R$ contains a terminal, then $|A| = 1$;
if $|A| = 0$, this means that $R$ contains no Steiner point.
The value of a subproblem $(R, A, \Pi)$
is the minimum weight of an $(m, r)$-light graph $G$,
such that for $A \neq \emptyset$,
\begin{itemize}
	\item $G$ connects all points $X \cap R$ to some point in $A$, and
	\item points in the same part of $\Pi$ are connected in $G$.
\end{itemize}
If $A = \emptyset$, then $X \subseteq R$ and the value is the 
minimum weight of an $(m, r)$-light graph $G$ such that
$G$ is contained inside $R$ and all terminals in $X$ are connected
in $G$.

	\section{Static Algorithms for Euclidean Steiner Tree}
\label{sec:offline}
In this section, we introduce an algorithm that answers the queries in the static setting, as stated in the following \Cref{thm:offline}. We will show how to make the algorithm fully-dynamic in the next section.

\begin{theorem}
	\label{thm:offline}
	There is an algorithm that,
	for every $0 < \epsilon < 1$, given as input an $n$-point dataset $X \subset \mathbb{R}^d$, 
	pre-processes $X$ and computes an implicit Steiner tree $T'$ for $X$ in 
	$O(n)\cdot (\log n)^{O(d^{2d}\epsilon^{-d})}\cdot 2^{O(4^d d^{4d} \epsilon^{-2d})}$ time such that $T'$ is $(1 + \epsilon)$-approximate with constant probability.
	After this pre-processing, the algorithm can deterministically report the weight of $T'$, report the root node of $T'$, answer membership query and answer neighbor query, in $O(1)$, $O(\frac{d^2}{\epsilon})^{d-1}$, $O(2^d\log n)$ and $(\log n)^{O(d^{2d}\epsilon^{-d})}\cdot 2^{O(4^d d^{4d} \epsilon^{-2d})}$ time, respectively. 
\end{theorem}

We impose the quadtree  decomposition with random shifts on the input terminal set $X$ as described in Section~\ref{sec:rev}.
Since there is no inherent parent-child relationship between nodes in a Steiner tree, the global root can be chosen arbitrarily.
We assume that the global root is selected based on a predefined priority assigned to all clusters. Specifically, for each level of the quadtree, a deterministic 
function is used to determine which cluster should contain the global root. The deterministic property ensures that we do not have the issue of failure probability after answering however many quries, and that the failure probability only needs to be analyzed for the error of the implicity Steiner tree

After we construct the quadtree, we can find an $(1+\epsilon)$-approximate solution by applying the original dynamic programming (DP) in Arora's approach,
since the original DP can find the optimal $(m, r)$-light solution and \Cref{thm:struct} ensures the existence of a solution that is $(m, r)$-light and $(1+\epsilon)$-approximate with $\frac{1}{2}$ probability.
Note that since $n=\Delta^d$ is an upper bound on the cardinality of $X$, we analyze the time complexity in terms of $n$. Also the depth of the quadtree is $O(\log \Delta)$, hence for simplicity, we slightly relax the bounds by using $O(\log n)$ as the depth.  

As mentioned, to answer the other types of queries, our plan is to augment the original DP of Arora
such that each DP entry is indexed by extra objects and each entry stores a small data structure in addition to a value.
The focus is to analyze the running time of computing the augmented DP, as well as the correctness and the time complexity of answering each type of these queries.
Then, \Cref{thm:offline} would follow from all of these components.

\paragraph{\textbf{Augmented DP for Steiner Tree.}}
Each entry in the augmented DP is indexed by $(R, A, \Pi, \varphi)$,
where $\varphi$ contains additional information for each part $S$ in $\Pi$ (recalling that $\Pi$ is the partition of $A$). 
Specifically, $\varphi(S)$ specifies
(i) whether the connected component containing $S$ inside $R$ includes
the global root of the entire Steiner tree, and (ii) the node in $S$ that is the closest to the global root, which we shall call the \emph{root} of $S$, denoted as $r(S)$.
In this context, closeness is defined as the number of edges between a point and the global root.
It is worth noting that, since the root of $S$, namely $r(S)$, is the point that is closest to the global root,
it ensures that all points in the connected component containing $S$ (inside $R$)
must connect to the global root via $r(S)$.

\paragraph{\textbf{Data Structure in Each Entry.}}
For a DP entry indexed by $(R, A, \Pi, \varphi)$, in addition to a value
as in the original DP, the following information is stored:
if $R$ is not a leaf cluster,
then for each child cluster $R_j$ (where $j \in [2^d]$) of $R$, there is
a pointer to some DP entry $(R_j, A_j, \Pi_j, \varphi_j)$ in the DP entry.

\paragraph{\textbf{Consistent Conditions.}}
Moreover, there is a collection $E$ ($E$ would not be stored) of directed edges
with vertex set in $\cup_{j} A_j$ that is \emph{consistent} 
in the following sense.
The collection of entries 
$\{(R_j, A_j, \Pi_j, \varphi_j): j \in [2^d]\}$
and the directed edges collection $E$ is consistent with $(R, A, \Pi, \varphi)$ if the following conditions hold.

\begin{itemize}
	
	\item If we consider each part in $\cup_j \Pi_j$ as a supernode, the edges $E$ will form a forest graph among these supernodes. In this graph, an edge in $E$ must connect components originating from different child clusters, meaning it connects distinct active portals from different child clusters.
	For the purpose of checking for no cycle, we assume that there is an edge
	between two parts for every common portal they share.
	Two parts from two different child clusters can share a common portal that
	lies on the common face between the two child clusters. The edges in $E$, along with the edges connecting common portals, will collectively form tree graphs among these supernodes as defined earlier.
	
	\item There can be at most one part $S \in \cup_j \Pi_j$ that contains
	the global root.
	
	\item The direction of an edge in $E$ is interpreted as going
	from a parent to a child.  For any part $S \in \cup_j \Pi_j$,
	the root of $S$ can have both outgoing and incoming edges in $E$,
	while all other portals in $S$ can only have outgoing edges in $E$.
	
	\item Several parts in $\cup_j \Pi_j$ are merged to become
	a single part in $\Pi$ (the portals lying on the common faces would be excluded since they are no longer portals for parts in $\Pi$).  This is because either an edge in $E$
	connects two different parts, or there is a common portal between two parts.
	
	\item If $a \in \cup_j A_j$ and $a$ is not a portal of $R$,
	then either $a$ is incident to some edge in $E$ or $a$ is a common 
	portal between two parts from different child clusters.
	
	If $a$ is also a portal of $R$,
	then $a$ must be in $A$ if $a$ is not incident to any edges in $E$,
	otherwise $a$ may or may not appear in $A$.
	
	\item The partition $\Pi$ of $A$ must be consistent with
	the aforementioned merging process.  (Recall that when
{\tiny }	several parts are merged, the direction of the edges in $E$
	must be consistent with the information in the $\varphi_j$'s.)
	
	Moreover, for the new part $S \in \Pi$,
	whether $S$ contains the global root or which portal is the root of $S$
	must be consistent with $\varphi$.
	
	\ignore{\item The value stored at $(R, A, \Pi, \varphi)$
		equals the sum of the values in $\{(R_j, A_j, \Pi_j, \varphi_j): j \in [2^d]\}$
		plus the weights of the edges in $E$. \wangbo{the value of a problem is the minimum over all consistent combinations of subproblems and $E$, however, to consider consistent combinations, the value of the problem should be determined. hence I think we should remove this item.}}
\end{itemize}

In what follows, we use the same $m$ and $r$ as guaranteed by Theorem~\ref{thm:struct} for the augmented DP for Steiner tree.

\begin{lemma}[Time Complexity of Computing Augmented DP]\label{lem:t_dp}
	The augmented DP for Steiner tree can be solved with time complexity $O(n)\cdot (\log n)^{O(d^{2d}\epsilon^{-d})}\cdot 2^{O(4^d d^{4d} \epsilon^{-2d})}$.
\end{lemma}

\begin{lemma}[Weight of the Steiner Tree]\label{lem:av_query}
	After solving the augmented DP for Steiner tree, we can report the weight of the implicit Steiner tree with time complexity $O(1)$, and the value is $(1+\epsilon)$-approximate with probability $\frac{1}{2}$.
\end{lemma}
\ignore{
\begin{proof}
	To obtain the weight of the implicit Steiner tree, we can simply access the entry corresponding to $R_{\log \Delta}$ in the DP table. The time complexity can be $O(1)$ if we are using a hash table.
	By \Cref{thm:struct}, we know that there exsits a solution that is $(m, r)$-light and $(1+\epsilon)$-approximates the optimal Steiner tree with probability at least $\frac{1}{2}$. The original DP can find the optimal $(m, r)$-light solution and can therefore find an $(1+\epsilon)$-approximation solution with probability at least $\frac{1}{2}$.
	Note that the augmented information $\varphi$ does not affect the computation of the approximate value since all valid solutions to the original DP would still be valid to the augmented DP (though some $\varphi$ would be added in the entries) and we can find counterparts for all valid solutions to the augmented DP in the solutions to the original DP.
	\qed
\end{proof}
}

\begin{lemma}[Global Root of the Steiner Tree]\label{lem:gr_query}
	After solving the augmented DP for Steiner Tree, we can report the global root of the Steiner tree with time complexity $O(\frac{d^2}{\epsilon})^{d-1}$.
\end{lemma}
\begin{proof}
	Let $\mathcal{I}$ be the entry whose associated cluster is $R_{\log \Delta}$ in the final solution. 
	$\mathcal{I}$ would point to $2^d$ subproblems $\mathcal{I'}=\{(R_i, A_i, \Pi_i, \varphi_i)\}_{i=1}^{2^d}$.
	The global root would belong to exactly one part $S\in \cup_{i=1}^{2^d}\Pi_i$. By checking $\varphi_i(S)$ such that $S\in \Pi_i$ for every $i\in\{1, 2, \dots, 2^d\}$, exactly one part would report ``Yes'' and the root of $S$ would be the global root.
	
	There are $2^d$ subproblems, each subproblem corresponds to a partition and each partition has $O(r)=O(\frac{d^2}{\epsilon})^{d-1}$ parts. We can check $\varphi$ of each part in $O(1)$, hence the total running time would be $2^d\cdot O(\frac{d^2}{\epsilon})^{d-1}=O(\frac{d^2}{\epsilon})^{d-1}$.
	\qed
\end{proof}

\begin{lemma}[Membership Queries]\label{lem:mem_query}
	After solving the augmented DP for Steiner Tree, we can correctly answer membership queries. The time complexity for answering the query is $O(2^d\log n)$. 
\end{lemma}

\ignore{
\begin{proof}
	Let $\mathcal{I}_0$ be the entry whose associated cluster is $R_{\log \Delta}$ in the final solution. Starting from $\mathcal{I}_0$, we can achieve the goal using BFS taking advantage of a lookup table that records all decisions made while solving the DP, i.e., each subproblem has $2^d$ pointers to its $2^d$ corresponding smaller subproblems and a pointer to its parent subproblem.
	Note that the information contained in $\varphi$ is implicitly used when using the lookup table, as the subproblems must satisfy consistent conditions in order to be included in the table.
	Specifically, we initialize a queue $P$ and push $\mathcal{I}_0$ to $P$. Every time we pop an element from $P$, denote the popped element as $\mathcal{I}=(R, A, \Pi, \varphi)$. We check if $u$ belongs to $A$ in $O(r)$ time: if yes, then $u$ is a node in the Steiner tree hence return ``Yes''; otherwise, we push all smaller subproblems of $\mathcal{I}$ that contain $u$ in its associated cluster to $P$ and repeat the above process. If ``Yes'' is not returned in the process, return ``No''. Note that there may be at most $2^d$ such subproblems since in each level there are at most $2^d$ clusters containing $u$. Therefore, there are $O(2^d\log n)$ clusters to consider. Hence the total running time for answering a membership query is $O(2^d \log n)$.  
	
	\qed
	
\end{proof}

}

\begin{lemma}[Neighbor Queries]\label{lem:pc_query}
	After solving the augmented DP for Steiner Tree, we can correctly answer the neighbor queries. The time complexity for answering the query is $O(\log n)^{O(d^{2d}\epsilon^{-d})}\cdot 2^{O(4^d d^{4d} \epsilon^{-2d})}$.
\end{lemma}

	\section{Dynamic Algorithm}\label{sec:dynamic}
In Section~\ref{sec:offline}, we have introduced an algorithm that can answer all types of queries in the static setting. In this section, we will show how to make the algorithm fully-dynamic. Specifically, we propose an algorithm that can efficiently maintain a data structure that supports the insertions/deletions of point as well as the afore-mentioned queries in the dynamic setting,
as formally stated in the following \Cref{thm:main}.
As we mention, we consider update sequences whose data points always lie on a discrete grid $[\Delta]^d$ at any time step.
Let $n := [\Delta]^d$.
This $n$ is an upper bound for the number of points in the active point set at any time step,
and hence we can use $n$ as a parameter for measuring the complexity.

\begin{restatable}[]{theorem}{mainthm}
	\label{thm:main}
	There exists a fully-dynamic algorithm that for any $0 < \epsilon, \delta < 1$, and any operation sequence $\sigma$
	whose underlying dataset belongs to $[\Delta]^d$ at any time step,
	maintains an implicit Steiner tree along with its weight and its root node.
	This tree is maintained in amortized
	$ O( \log \frac{t}{\delta})  \cdot (\log n)^{O(d^{2d}\epsilon^{-d})}\cdot 2^{O(4^d d^{4d} \epsilon^{-2d})} $ time for the $t$-th operation ($\forall t \geq 1$),
	and with probability at least $1 - \delta$, it is $(1 + \epsilon)$-approximate at all time steps simultaneously.
	The subroutine for answering membership and neighbor queries are deterministic, and run in time $O(2^d\log n)$ and $(\log n)^{O(d^{2d}\epsilon^{-d})}\cdot 2^{O(4^d d^{4d} \epsilon^{-2d})}$, respectively.
\end{restatable}


Note that we have already shown in Theorem~\ref{thm:offline} that all these queries can be answered efficiently after computing the augmented DP.
Hence, our proof of Theorem~\ref{thm:main} aims to maintain the augmented DP efficiently instead of computing it from scratch after each update.

We mostly focus on the simpler version where a single update can be handled with constant success probability.
Then to make the algorithm handle all updates and succeed for all of them simultaneously,
a standard amplification is to repeat the algorithm for $\log(1 / \delta)$ times, so that the failure probability is reduced to $\delta$.
However, the operation sequence can be indefinitely long so it is not possible to set a target failure probability $\delta$ in advance,
so that the total failure probability is still bounded.
Our solution is to rebuild the data structure with a decreased target failure probability whenever sufficiently many updates are preformed. 

Next, we would introduce a dynamic data structure wrapping the static algorithm described in Section~\ref{sec:offline} that supports updates.

\paragraph{\textbf{Data Structure for Dynamic Algorithms.}}
Let $\Gamma_{\epsilon}$ be a data structure storing a set of tuples that correspond to all subproblems in the augmented dynamic programming described in Section~\ref{sec:offline}. We use the same $m$ and $r$ as guaranteed by Theorem~\ref{thm:struct} (however, note that the parameters may be different from previous sections since the size of bounding box has changed), i.e., in a tuple $(R, A, \Pi, \varphi)$, $A$ is a size of $r := O(\frac{d^{2}}{\epsilon})^{d-1}$ subset of the $m := O(d^{\frac{3}{2}} \cdot \frac{\log n}{\epsilon})^{d-1}$ portals of $R$.
$\Gamma_{\epsilon}$ supports the following operations:
\begin{itemize}
	\item \textit{Initialize().}
	Conduct the tree decomposition with random shifts on the bounded area. Then we use the augmented dynamic programming described in Section~\ref{sec:offline} to find the optimal $(m, r)$-light solution. All subproblems will be identified as a set of tuples $(R, A, \Pi, \varphi)$ and will be stored together with their evaluated values. Note each subproblem will have optimal value $0$ since there are no points when initializing.
	\item \textit{Insert($x$).}
	When we call Insert($x$), we find the set of all clusters $\mathcal{R}^x\subset\mathcal{R}$ that contain $x$. There will be exactly a cluster containing $x$ in each level. Let $R_i^x\in \mathcal{R}^x$ be the level-$i$ cluster that contains $x$. Note that when a point $x$ is inserted, if it is not the new root of the Steiner Tree, then only the subproblems associated with clusters in $\mathcal{R}^x$ will be affected.
	If $x$ is the new root, let $y$ be the old root, then the subproblems associated with clusters in $\mathcal{R}^x \cup \mathcal{R}^y$ will be affected. 
	We will update the values of these subproblems in a bottom-up manner. Specifically speaking, the subproblems associated with $R_0^{k}$ ($k$ can be either $x$ or $y$) will be updated trivially since $k$ is the only point in the cluster; for the subproblems $\mathcal{I} = (R_i^k, A, \Pi, \varphi)$, where $i\in\{1,2,\ldots, \log \Delta\}$, we consider all combinations of entry indices $I'$ and directed edges $E'$ that are consistent with $\mathcal{I}$, and update the value of $\mathcal{I}$ to the minimum sum of the DP values of $I's$ and the weight of $E'$ over those combinations.
	\item \textit{Delete($x$).}
	Deletion procedure is similar to that of insertion. We need to update the values of all subproblems whose corresponding clusters contain $x$ in a bottom-up approach.
	If $x$ is the root, let $y$ be the new root. We need to update the values of all subproblems whose corresponding clusters contain $y$ in a bottom-up approach.
\end{itemize}

\paragraph{\textbf{Pre-Computation.}} Note that the initialization of $\Gamma_{\epsilon}$ is independent of $\sigma$. Hence, we can construct enough copies of $\Gamma_{\epsilon}$ by calling \textit{Initialize}() before the start of $\sigma$ and this part of cost would not count in the time complexity analysis for each update.

\begin{lemma}[The Correctness of the Data Structure]
	\label{lem:ds_correct}
	Let $0 < \epsilon < 1$, $\Gamma_{\epsilon}$ can maintain an implicit Steiner tree that is $1+\epsilon$-approximate with probability at least $\frac{1}{2}$, and report the weight $w$ and the global root of the implicit Steiner tree.
	The time complexity for $\Gamma_{\epsilon}$ to finish a single operation of \textit{Initialize($Z$)} is $O(n)\cdot (\log n)^{O(d^{2d}\epsilon^{-d})}\cdot 2^{O(4^d d^{4d} \epsilon^{-2d})}$; and the time complexity to finish a single operation of \textit{Insert($x$)} or \textit{Delete($x$)} is $(\log n)^{O(d^{2d}\epsilon^{-d})}\cdot 2^{O(4^d d^{4d} \epsilon^{-2d})}$.
\end{lemma}

\begin{lemma}[Answering the Queries]
	\label{lem:dyn_queries}
	By maintaining a copy of $\Gamma_{\epsilon}$ we can, for any time step $t$,  
	maintain the weight and the global root of an implicit Steiner tree that is $(1+\epsilon)$-approximation with probability at least $\frac{1}{2}$, and can correctly answer the membership queries and the neighbor queries about the implicit Steiner tree, at any time step $t$.
	The time complexity for answering the two types of queries are $O(2^d\log n)$ and $(\log n)^{O(d^{2d}\epsilon^{-d})}\cdot 2^{O(4^d d^{4d} \epsilon^{-2d})}$, separately.
\end{lemma}
\begin{proof}
	
	
	We can construct a copy of $\Gamma_{\epsilon}$ and call \textit{Initialize} to initialize. Then for each update call \textit{Insert} or \textit{Delete}. According to Lemma~\ref{lem:ds_correct}, the implicit Steiner tree maintained by $\Gamma_{\epsilon}$ is $(1+\epsilon)$-approximate with probability at least $\frac{1}{2}$, and we can report the weight and the global root of the implicit Steiner tree.
	
	
	According to Lemma~\ref{lem:mem_query} and Lemma~\ref{lem:pc_query}, we can correctly answer membership queries and parent-child queries by using those methods on $\Gamma_{\epsilon}$ with time complexity $O(2^d\log n)$ and $(\log n)^{O(d^{2d}\epsilon^{-d})}\cdot 2^{O(4^d d^{4d} \epsilon^{-2d})}$, separately.
	\qed
\end{proof}

\begin{lemma}[High Probability Guarantee by Repetition]\label{lem:repetition}
	For any $0<\epsilon,~\delta<1$ and any operation sequence $\sigma$ with length $T$, we can maintain an implicit Steiner tree that is $(1+\epsilon)$-approximate with probability at least $1-\delta$ along with its weight and global root, for all time steps simultaneously by maintaining $\log \frac{T}{\delta}$ independent copies of $\Gamma_{\epsilon}$. The time cost for each update is $O(\log \frac{T}{\delta}) \cdot (\log n)^{O(d^{2d}\epsilon^{-d})}\cdot 2^{O(4^d d^{4d} \epsilon^{-2d})}$. 
	We can also correctly answer the membership queries and the neighbor queries with probability $1$, at any time step $t$, with time complexity $O(2^d\log n)$ and $(\log n)^{O(d^{2d}\epsilon^{-d})}\cdot 2^{O(4^d d^{4d} \epsilon^{-2d})}$, separately.
\end{lemma}
\begin{proof}
	We can maintain an $(1+\epsilon)$-approximation to the optimal value with probability at least $\frac{1}{2}$ by one copy of $\Gamma_{\epsilon}$ according to Lemma~\ref{lem:dyn_queries}. Having $\log \frac{T}{\delta}$ independent copies can then guarantee that with probability at least $1-(\frac{1}{2})^{\log \frac{T}{\delta}} = 1-\frac{\delta}{T}$ at least one copy can maintain an $(1+\epsilon)$-approximation to the optimal value. Then by a union bound, the probability that all time steps succeed simultaneously is $1-\delta$.
	Note that each approximate value corresponds to a feasible solution to the Euclidean Steiner tree problem, we hence fix the copy $\tilde{\Gamma}$ with the minimum approximate value to perform further queries.
	
	The time cost comes from two parts: the update (\textit{Insert} and \textit{Delete}), and finding $\tilde{\Gamma}$ after each update. The initialization of $O(\log \frac{T}{\delta})$ copies in the pre-computation. According to Lemma~\ref{lem:ds_correct}, each update takes $(\log n)^{O(d^{2d}\epsilon^{-d})}\cdot 2^{O(4^d d^{4d} \epsilon^{-2d})}$. There are $\log \frac{T}{\delta}$ copies in the algorithm, hence each update would take $\log \frac{T}{\delta}\cdot (\log n)^{O(d^{2d}\epsilon^{-d})}\cdot 2^{O(4^d d^{4d} \epsilon^{-2d})}$ time. After each update, according to Lemma~\ref{lem:av_query}, we need $O(\log \frac{T}{\delta})$ time to find $\tilde{\Gamma}$.
	
	Moreover, a global root is maintained in $\tilde{\Gamma}$ according to Lemma~\ref{lem:dyn_queries}, and we can also correctly answer the queries by performing those queries on $\tilde{\Gamma}$. Hence the correctness and time complexity for other queries are as desired.
	\qed
\end{proof}

\begin{proof}[Proof of \Cref{thm:main}]
	According to Lemma~\ref{lem:dyn_queries}, we can maintain an implicit Steiner tree that is $(1+\epsilon)$ with probability at least $\frac{1}{2}$, report the weight and the global root of the implicit Steiner tree, and correctly answer the membership queries and the parent-child queries, at any time step $t$. The time complexity for answering the two types of queries are $O(2^d\log n)$ and $(\log n)^{O(d^{2d}\epsilon^{-d})}\cdot 2^{O(4^d d^{4d} \epsilon^{-2d})}$, separately.
	
	Then we analyze the per time step cost. The time cost comes from the updates (\textit{Insert} and \textit{Delete}) (note that \textit{Initialize} is done in pre-computation). According to Lemma~\ref{lem:ds_correct}, each update would take $(\log n)^{O(d^{2d}\epsilon^{-d})}\cdot 2^{O(4^d d^{4d} \epsilon^{-2d})}$ time.

	Partition $\sigma$ into phases with different lengths indexed from $1$. Specifically, the $i$-th phase would have $2^i n$ time steps. For the $i$-th phase, we would guarantee that we can maintain an implicit Steiner tree that with probability at least $1-\frac{\delta}{2^i}$ is $(1+\epsilon)$-approximate. By a union bound on the phases, we can guarantee that all time steps succeed simultaneously with probability at least $1-\sum_i \frac{\delta}{2^i}\geq 1-\delta$. Then we need to maintain $\log \frac{2^i n}{\delta}=O(i + \log n + \log \frac{1}{\delta})$ copies for the $i$-th phase according to a similar analysis as in the proof of Lemma~\ref{lem:repetition}.
	
	Note that we do not need to reconstruct all copies at the beginning of each phase, instead, we just need to add some new copies since Theorem~\ref{thm:struct} has no dependency on the time step. We will keep inserting all points in the active point set to the new copies for later updates.
	
	
	We first analyze the reconstruction cost for the copies. The time step $t$ would lie in the $O(\log \frac{t}{n})$-th phase and there are $O(2^{\log \frac{t}{n}} n) = O(t)$ time steps in the phase. The total number of copies of $\Gamma_\epsilon$ needed for the first $t$ steps is $O(\log \frac{t}{\delta})$. For each copy, we may need to insert at most $n$ points for reconstruction. Hence the total cost for this part is $n\cdot O(\log \frac{t}{\delta}) \cdot (\log n)^{O(d^{2d}\epsilon^{-d})}\cdot 2^{O(4^d d^{4d} \epsilon^{-2d})}$.
	
	Then we analyze the update cost for the first $t$ steps. Since there are $\log \frac{2^i n}{\delta}=O(i + \log n + \log \frac{1}{\delta})$ copies in the $i$-th phase, hence the total update cost would be $\sum_{i=1}^{\log \frac{t}{n}} O((i + \log n + \log \frac{1}{\delta})\cdot (2^i) \cdot n)\cdot (\log n)^{O(d^{2d}\epsilon^{-d})}\cdot 2^{O(4^d d^{4d} \epsilon^{-2d})} = O( t\log \frac{t}{\delta}\cdot (\log n)^{O(d^{2d}\epsilon^{-d})}\cdot 2^{O(4^d d^{4d} \epsilon^{-2d})}$. 
	
	Note that $t>n$. Therefore, the amortized time cost for the time step $t$ is $\log\frac{t}{\delta} \cdot (\log n)^{O(d^{2d}\epsilon^{-d})}\cdot 2^{O(4^d d^{4d} \epsilon^{-2d})}$.
	
	The correctness and time complexity for other queries are also as desired according to Lemma~\ref{lem:repetition}. 
	\qed
\end{proof}

	\bibliographystyle{splncs04}
	\bibliography{ref}

\begin{thebibliography}{10}
\providecommand{\url}[1]{\texttt{#1}}
\providecommand{\urlprefix}{URL }
\providecommand{\doi}[1]{https://doi.org/#1}

\bibitem{DBLP:journals/jacm/Arora98}
Arora, S.: Polynomial time approximation schemes for euclidean traveling
  salesman and other geometric problems. J. {ACM}  \textbf{45}(5),  753--782
  (1998)

\bibitem{berman1994improved}
Berman, P., Ramaiyer, V.: Improved approximations for the steiner tree problem.
  Journal of Algorithms  \textbf{17}(3),  381--408 (1994)

\bibitem{bern1989steiner}
Bern, M., Plassmann, P.: The steiner problem with edge lengths 1 and 2.
  Information Processing Letters  \textbf{32}(4),  171--176 (1989)

\bibitem{byrka2010improved}
Byrka, J., Grandoni, F., Rothvo{\ss}, T., Sanita, L.: An improved lp-based
  approximation for steiner tree. In: Proceedings of the forty-second ACM
  symposium on Theory of computing. pp. 583--592 (2010)

\bibitem{arxivFullVersion}
Chan, T.H.H., Goranci, G., Jiang, S.H.C., Wang, B., Xue, Q.: Fully dynamic
  algorithms for euclidean steiner tree. CoRR  \textbf{arxiv} (2023)

\bibitem{cheng2013steiner}
Cheng, X., Du, D.Z.: Steiner trees in industry, vol.~11. Springer Science \&
  Business Media (2013)

\bibitem{DBLP:journals/dcg/Eppstein95}
Eppstein, D.: Dynamic euclidean minimum spanning trees and extrema of binary
  functions. Discret. Comput. Geom.  \textbf{13},  111--122 (1995)

\bibitem{garey1977complexity}
Garey, M.R., Graham, R.L., Johnson, D.S.: The complexity of computing steiner
  minimal trees. SIAM journal on applied mathematics  \textbf{32}(4),  835--859
  (1977)

\bibitem{gu2013power}
Gu, A., Gupta, A., Kumar, A.: The power of deferral: maintaining a
  constant-competitive steiner tree online. In: Proceedings of the forty-fifth
  annual ACM symposium on Theory of Computing. pp. 525--534 (2013)

\bibitem{gupta2014online}
Gupta, A., Kumar, A.: Online steiner tree with deletions. In: Proceedings of
  the twenty-fifth annual ACM-SIAM symposium on Discrete algorithms. pp.
  455--467. SIAM (2014)

\bibitem{held2011combinatorial}
Held, S., Korte, B., Rautenbach, D., Vygen, J.: Combinatorial optimization in
  vlsi design. Combinatorial Optimization pp. 33--96 (2011)

\bibitem{DBLP:journals/jacm/HolmLT01}
Holm, J., de~Lichtenberg, K., Thorup, M.: Poly-logarithmic deterministic
  fully-dynamic algorithms for connectivity, minimum spanning tree, 2-edge, and
  biconnectivity. J. {ACM}  \textbf{48}(4),  723--760 (2001)

\bibitem{hougardy19991}
Hougardy, S., Pr{\"o}mel, H.J.: A 1.598 approximation algorithm for the steiner
  problem in graphs. In: Proceedings of the tenth annual ACM-SIAM symposium on
  Discrete algorithms. pp. 448--453 (1999)

\bibitem{hwang1992steiner}
Hwang, F.K., Richards, D.S.: Steiner tree problems. Networks  \textbf{22}(1),
  55--89 (1992)

\bibitem{imase1991dynamic}
Imase, M., Waxman, B.M.: Dynamic steiner tree problem. SIAM Journal on Discrete
  Mathematics  \textbf{4}(3),  369--384 (1991)

\bibitem{karpinski1997new}
Karpinski, M., Zelikovsky, A.: New approximation algorithms for the steiner
  tree problems. Journal of Combinatorial Optimization  \textbf{1},  47--65
  (1997)

\bibitem{lkacki2015power}
Lacki, J., Ocwieja, J., Pilipczuk, M., Sankowski, P., Zych, A.: The power of
  dynamic distance oracles: Efficient dynamic algorithms for the steiner tree.
  In: {STOC}. pp. 11--20. {ACM} (2015)

\bibitem{matsuyama1980approximate}
Matsuyama, A.: An approximate solution for the steiner problem in graphs. Math.
  Japonica  \textbf{24},  573--577 (1980)

\bibitem{megow2016power}
Megow, N., Skutella, M., Verschae, J., Wiese, A.: The power of recourse for
  online mst and tsp. SIAM Journal on Computing  \textbf{45}(3),  859--880
  (2016)

\bibitem{DBLP:journals/siamcomp/Mitchell99}
Mitchell, J.S.B.: Guillotine subdivisions approximate polygonal subdivisions:
  {A} simple polynomial-time approximation scheme for geometric tsp, k-mst, and
  related problems. {SIAM} J. Comput.  \textbf{28}(4),  1298--1309 (1999)

\bibitem{promel1997rnc}
Pr{\"{o}}mel, H.J., Steger, A.: Rnc-approximation algorithms for the steiner
  problem. In: {STACS}. Lecture Notes in Computer Science, vol.~1200, pp.
  559--570. Springer (1997)

\bibitem{robins2005tighter}
Robins, G., Zelikovsky, A.: Tighter bounds for graph steiner tree
  approximation. SIAM Journal on Discrete Mathematics  \textbf{19}(1),
  122--134 (2005)

\bibitem{DBLP:journals/siamdm/RubinsteinTW97}
Rubinstein, J.H., Thomas, D.A., Wormald, N.C.: Steiner trees for terminals
  constrained to curves. {SIAM} J. Discret. Math.  \textbf{10}(1),  1--17
  (1997)

\bibitem{zelikovsky1996better}
Zelikovsky, A.: Better approximation bounds for the network and euclidean
  steiner tree problems. Tech. rep., Technical Report CS-96-06, Department of
  Computer Science, University of~ (1996)

\bibitem{zelikovsky199311}
Zelikovsky, A.Z.: An 11/6-approximation algorithm for the network steiner
  problem. Algorithmica  \textbf{9},  463--470 (1993)

\end{thebibliography}
	
	
\end{document}